\documentclass[letterpaper, 10 pt, conference]{ieeeconf}  

\IEEEoverridecommandlockouts                              

\usepackage{amsmath} 
\usepackage{amssymb}  

\usepackage{enumerate}
\usepackage{mathtools}

\usepackage{xcolor}

\usepackage{booktabs}

\newtheorem{lemma}{\bf Lemma}
\newtheorem{theorem}{\bf Theorem}

\newtheorem{assumption}{\bf Assumption}
\newtheorem{corollary}{\bf Corollary}

\usepackage{acronym}
\acrodef{LTI}{linear-time-invariant}
\acrodef{LFT}{linear fractional transformation}
\acrodef{SDP}{semidefinite program}
\acrodef{QMI}{quadratic matrix inequality}

\usepackage{tikz}

\usetikzlibrary{patterns,decorations.pathreplacing}
\usepackage{pgfplots}
\pgfplotsset{compat=1.18}

\usepackage{cite}

\title{\LARGE \bf
Data-driven Estimator Synthesis with Instantaneous Noise
}

\author{Felix Brändle and Frank Allgöwer
\thanks{F. Allg\"ower is thankful that this work was funded by the Deutsche Forschungsgemeinschaft (DFG, German Research Foundation) under Germany’s Excellence Strategy -- EXC 2075 -- 390740016 and within grant AL
	316/15-1 – 468094890. F. Br\"andle thanks the International Max Planck Research School for Intelligent Systems (IMPRS-IS) for supporting him.
	}
\thanks{F. Br\"andle and F. Allg\"ower are with the University of Stuttgart, Institute for Systems
		Theory and Automatic Control, 70550 Stuttgart,
		Germany. (e-mail: \{felix.braendle, frank.allgower\}@ist.uni-stuttgart.de)}%
}

\begin{document}
	
\maketitle
\thispagestyle{empty}
\pagestyle{empty}

\begin{abstract}
Data-driven controller design based on data informativity has gained popularity due to its straightforward applicability, while providing rigorous guarantees.
However, applying this framework to the estimator synthesis problem introduces technical challenges, which can only be solved so far by adding restrictive assumptions.
In this work, we remove these restrictions to improve performance guarantees.
Moreover, our parameterization allows the integration of additional structural knowledge, such as bounds on parameters.
Our findings are validated using numerical examples.
\end{abstract}

\section{Introduction}
In recent years, data-driven control methods have increased in popularity \cite{Willems2005,Berberich2020, Coulson2019,Martin2023a}.
Model-based techniques necessitate an accurate model, which can be difficult to obtain and often requires expert knowledge.
In contrast, data-driven methods bypass this step by directly synthesizing a controller from data, eliminating the need for a model.

A common approach is set-membership estimation \cite{Cerone2012}.
Instead of identifying a single model, this method identifies the set of all systems, that could have generated the available data.
Since the true system must lie within this set, guarantees provided for all systems in this set must also hold for the true system.
However, the complexity of this method grows rapidly with increasing number of samples, necessitating approximation techniques \cite{Castano2011}.
To address this issue, an established alternative is the data-informativity approach \cite{Waarde2020a, Waarde2022}.
This framework considers a linear regression model
\begin{equation}
	Y = \theta X + W
\end{equation}
with known regressor $X$, and regressand $Y$, and unknown parameter matrix $\theta$, and disturbance $W$.
When multiple samples $X$ and $Y$ are collected with $W$ being independent between each sample, this is called the instantaneous noise \cite{Waarde2022}.
Dependence between samples may arise, if the data is generated by a dynamical system and a combined bound on energy of the disturbance $W$ is assumed \cite{Braendle2024}.
Under some general assumptions, it is possible to recreate a tight description of all possible $\theta$ consistent with the data by a \ac{QMI} 
\begin{equation}
	\begin{bmatrix}
		\theta^\top \\ I
	\end{bmatrix}^\top
	\Phi
	\begin{bmatrix}
		\theta^\top \\ I
	\end{bmatrix} \succeq 0. \label{eq:Intro:QMI}
\end{equation}
Robust control offers many methods to synthesize robust state feedback controllers for such uncertainty descriptions in $\theta^\top$ by solving a \ac{SDP} \cite{Scherer2000}.

Next to state feedback controllers, observer and estimator synthesis is also an important problem class.
For example, moving horizon estimation estimates the state by repeatedly solving an optimization problem.
In order to incorporate data, the system dynamics are replaced by Hankel-matrices and additional regularization terms are added to account for disturbances in the data collection \cite{Wolff2022}.
Estimators based on Luenberger observers can also be employed.
However, this requires the true state transition matrix, which is not available if the data is affected by disturbances  \cite{Mishra2022}. 
In \cite{Liu2023a}, the authors address this by including an approximation of the true state transition matrix and additional regularization to synthesize a stabilizing LQG-controller.

In this work, we consider a different approach by deriving a \ac{LFT}-formulation to apply robust control methods.
This allows us to reformulate the estimator design in terms of designing a general dynamic output feedback controller not requiring an approximation of the true state transition matrix anymore.
Moreover, we can optimize the $\mathcal{H}_\infty$-norm to minimize the effects of any disturbance on the estimation.
To do so, we consider a similar setup as in the data-informativity framework with instantaneous noise.
As our main contribution, we derive a new \ac{QMI}-based parameterization as in \eqref{eq:Intro:QMI}, but for $\theta$ instead of $\theta^\top$ to synthesize an estimator.
This leads to a convex \ac{SDP}, which can be solved efficiently and at the same time allow for flexible uncertainty characterizations to reduce conservatism.
Furthermore, we can impose additional structural knowledge, such as bounds on individual parameter ranges.
In addition, we compare our parameterization to the data-informativity framework, highlighting benefits of each method.
Finally, we validate our findings through numerical experiments.

\emph{Notation}:
We denote the $n\times n$ identity matrix and the $p\times q$ zero matrix as $I_n$ and $0_{p\times q}$, respectively, where we omit the indices if the dimensions are clear from context. 
For the sake of space limitations, we use $[\star] $, if the corresponding matrix can be inferred from symmetry of the whole expression. 
Moreover, we use $P\succ 0$ ($P\succeq 0$), if the symmetric matrix $P$ is positive (semi-) definite. 
Negative (semi-) definiteness is denoted by $P\prec 0$ ($P\preceq 0$). 
We write $\sigma_{\mathrm{min}}(A)$ and $\sigma_{\mathrm{max}}(A)$ for the minimal and maximal singular value of $A$. 
Furthermore, we use $\mathrm{diag}(A_1,\ldots,A_k)$ to abbreviate a block diagonal matrix with $A_i\,i=1,\ldots,k$ as blocks.
Moreover, if $A$ has full rank, its Moore-Penrose inverse is denoted by $A^\dagger$.
We write $A^\perp$ for a matrix, which rows form an orthonormal basis of the kernel of $A$, i.e., $AA^{\perp\top} = 0$.

\section{Set-membership estimation} \label{sec:MainTheorem}
We consider a linear regression model of the form
\begin{equation}
	Y_k = \theta_{\mathrm{tr}}X_k + W_k \label{eq:Main:DataEq}
\end{equation}
with regressand $Y_k\in\mathbb{R}^{p\times m}$, regressor $X_k\in\mathbb{R}^{n\times m}$, unknown disturbance $W_k\in\mathbb{R}^{p\times m}$, and unknown parameter matrix $\theta_{\mathrm{tr}}\in\mathbb{R}^{p\times n}$.
We sample $N$ measurements $\{X_k,Y_k\}_{k=1}^{N}$ from this model affected by an unknown disturbance $W_k\in\mathcal{W}$ with
\begin{equation}
	\mathcal{W}\in\left\{W\in\mathbb{R}^{p\times m} \mid [\star]^\top \begin{bmatrix} -Q & 0 \\0 & R\end{bmatrix}\begin{bmatrix}	W -\bar{W}\\ I	\end{bmatrix} \succeq 0\right\} \label{eq:main:WSet},
\end{equation}
$0\preceq Q\in\mathbb{R}^{p \times p}$ and $0\preceq R\in\mathbb{R}^{m \times m}$.
Without loss of generality, we set  $\bar{W}$ to zero, as this offset can be absorbed in $Y_k$.
The overall goal is to find $N$ sets 
\begin{equation*}
	\Sigma_k = \left\{\theta\in\mathbb{R}^{p\times n} \mid \exists W\in\mathcal{W}:\: Y_k = \theta X_k + W \right\}
\end{equation*}
to derive the set of all $\theta$ consistent with the collected data by $\bigcap_{k=1}^{N} \Sigma_k$.
We represent each set $\Sigma_k$ by a \ac{QMI} in $\theta$. 
The intersection can be approximated by applying the S-Lemma.
To this end, we first derive a \ac{QMI}-based description of $\Sigma_k$.

\subsection{Single Data Point}
In this subsection, we consider a single data sample $\{X_1,Y_1\}$ with $X_1$ having full column rank.
For simplicity, we omit the index $1$.
Given the regression model 
\begin{equation}
	Y = \theta X + W, 
\end{equation}
we express $W$ in terms of $\theta$ and insert it in $\mathcal{W}$, leading to
\begin{equation*}
	\Sigma = \{ \theta \in\mathbb{R}^{p \times n} \mid 
	[\star]^\top
	\begin{bmatrix}
		-Q & 0 \\
		0 & R 
	\end{bmatrix}
	\begin{bmatrix} 
		\theta X - Y \\ 
		I	
	\end{bmatrix}
	\succeq 0\}.
\end{equation*}
This results in an unbounded set, since any $\theta$ satisfying $\theta X = 0$ is also contained in $\Sigma$.
To apply set-membership methods, we reformulate $\Sigma$ using a different \ac{QMI}.
The objective is to eliminate $X$ to obtain a description in terms of $\theta$ rather than $\theta X$.
We consider the parameterization
\begin{align*}
	&\hat{\Sigma} = \{\theta \in\mathbb{R}^{p \times n} \mid \\
	&[\star]^\top
	\begin{bmatrix}
		-Q & 0\\
		0  & G^\top(R+\hat{R})G + G_0^\top \hat{Q} G_0
	\end{bmatrix}
	\begin{bmatrix} 
		\theta - \hat{\theta} \\ 
		I 
	\end{bmatrix} \succeq 0
	\}
\end{align*}
with $G=X^\dagger$ and $G_0=(I-XG)$, where the matrices $0\prec \hat{R}\in\mathbb{R}^{m \times m}$, $0\preceq\hat{Q}\in\mathbb{R}^{n\times n}$, and $\hat{\theta}\in\mathbb{R}^{p \times n}$ are decision variables.
The set $\hat{\Sigma}$ is described by a \ac{QMI} centered around $\hat{\theta}$.
We note that $X$ has full column rank, allowing us to compute a left inverse $G=X^\dagger$ with $G X = I$.
The main idea is to apply the identity
\begin{equation}
	\begin{bmatrix}	\theta X \\ I\end{bmatrix} = 
	\begin{bmatrix}	\theta \\ G\end{bmatrix} X.
\end{equation}
on $\Sigma$ and relax the condition on the column space of $X$ to the whole $\mathbb{R}^n$, thus isolating $\theta$.
A simple multiplication yields the term $G^\top R G$ in $\hat{\Sigma}$.
To account the additional effects of $-\theta Q \theta$, we add the compensatory term $G^\top \hat{R} G + G_0^\top \hat{Q} G_0$.
The separation into $G$ and $G_0$ distinguishes between the column space of $X$ and the remaining $\mathbb{R}^n$.
This follows from the Moore-Penrose-Inverse as $G_0$ is a projection matrix on the complementary space of $X$.
\begin{assumption}\label{ass:main:prior}
	$\theta_{\mathrm{tr}}\in\Sigma_0$ with
	\begin{equation*}
		\Sigma_0 = \left\{ \theta \in\mathbb{R}^{p \times n} \mid 
		\begin{bmatrix} 
			\theta\\ 
			I	
		\end{bmatrix}^\top
		\Phi_{0,i}
		\begin{bmatrix} 
			\theta\\ 
			I	
		\end{bmatrix}
		\succeq 0 \quad \forall i = 1,\ldots,N_0 \right\}.
	\end{equation*}
\end{assumption}
Assumption\,\ref{ass:main:prior} restricts the $\theta$ we need to consider to prevent $\theta Q \theta$ from growing arbitrary large.
This assumption is quite restrictive, since it requires additional knowledge, which may not always be available. 
That is why, we present a method for constructing the prior $\Sigma_0$ in Section\,\ref{subsec:MultPoints}.
With this, we can state the main theorem of this paper and how to select $\hat{R}$, $\hat{Q}$, and $\hat{\theta}$ appropriately.
\begin{theorem} \label{Theo:main:MainTheo}
	Suppose $X\in\mathbb{R}^{n \times m}$ has full column rank and $\bar{\theta}\in\mathbb{R}^{p \times n}$ is an arbitrary matrix.
	If $0\prec \hat{R}\in\mathbb{R}^{m \times m}$ and $0\preceq\hat{Q}\in\mathbb{R}^{n\times n}$ satisfy
	\begin{equation}
		R \preceq \gamma^2 \hat{R} \label{eq:Main:RhatTheorem},
	\end{equation}
	\begin{equation}
		[\star]^\top 
		\begin{bmatrix}
			-(Q+\gamma^2Q) & 0 \\
			0 & \hat{Q}
		\end{bmatrix}
		\begin{bmatrix} 
			\theta - \bar{\theta} \\ 
			I 
		\end{bmatrix}
		G_0 \succeq 0 \quad \forall \theta \in\Sigma_0  \label{eq:Main:QhatTheorem}
	\end{equation}
	with $\hat{\theta}=YG + \bar{\theta}G_0$, then
	\begin{enumerate}[i)]
		\item $\theta \in (\Sigma \cap \Sigma_0) \implies \theta \in \hat{\Sigma}$, \label{item:Main:TheoSetInt}
		\item \label{item:Main:TheoNotSet} $
		[\star]^\top
		\begin{bmatrix}
			-Q & 0 \\
			0 & R + \hat{R}
		\end{bmatrix}
		\begin{bmatrix} 
			\theta X - Y \\ 
			I	
		\end{bmatrix}
		\nsucceq 0 \implies \theta\not\in\hat{\Sigma},
		$
		\item $\lim\limits_{\substack{\sigma_\mathrm{min}(\hat{Q})\to \infty \\ \sigma_\mathrm{max}(\hat{R})\to 0}} \hat{\Sigma} = \Sigma$. \label{item:Main:TheoConvergence}
	\end{enumerate}
\end{theorem}
\begin{proof}
	\ref{item:Main:TheoSetInt}) 
	First, we take an arbitrary vector $v$ and parameterize it according to $v=v_1 + v_2$ with $v_1 = XG v$ and $v_2=G_0v$.
	Since $X$ has full column rank, the Moore-Penrose inverse satisfies $GX=I$, $G_0X=0$, and $G_0G_0=G_0$.
	Moreover by, employing $v_1 = XGv$ and $\theta\in\Sigma$, we conclude
	\begin{equation*}
		v_1^\top	
		\begin{bmatrix} 
			\theta  - YG \\ 
			G	
		\end{bmatrix}^\top
		\begin{bmatrix}
			-Q & 0 \\
			0 &  R 
		\end{bmatrix}
		\begin{bmatrix} 
			\theta  - YG \\ 
			G	
		\end{bmatrix} v_1 \geq 0. 
	\end{equation*}
	Next, we use quadratic completion, $XGv_1=v_1$ and $\hat{R}\succ 0$ to imply invertibility and  existence of a positive definite matrix square root $\hat{R}^{\tfrac{1}{2}}$ such that
	\begin{equation*}
		\begin{split}
			&\|\hat{R}^{\tfrac{1}{2}}Gv_1 - \hat{R}^{-\tfrac{1}{2}}X^\top (\theta - \hat{\theta})^\top Q (\theta - \hat{\theta}) v_2 \|^2  \\
			&=v_1^\top G^\top \hat{R} G v_1 -2 v_1^\top (\theta - \hat{\theta})^\top Q (\theta - \hat{\theta}) v_2 \\ 
			&+v_2^\top(\theta - \hat{\theta})^\top Q (\theta - \hat{\theta})X\hat{R}^{-1}X^\top (\theta - \hat{\theta})^\top Q (\theta - \hat{\theta}) v_2.
		\end{split}
	\end{equation*}
	
	In order to show satisfaction of the \ac{QMI} in $\hat{\Sigma}$, we multiply an arbitrary vector $v$ and its transpose from the left and right, respectively, use $GG_0=0$ and lower bound it
	\begin{equation*}
		\begin{split}		
			&[\star]^\top
			\begin{bmatrix}
				-Q & 0 \\
				0 &  R 
			\end{bmatrix}
			\begin{bmatrix} 
				\theta  - YG \\ 
				G	
			\end{bmatrix} v_1 +
			v_1^\top G^\top \hat{R} G v_1+ 
			v_2^\top \hat{Q} v_2 \\
			&-2 v_1^\top (\theta-\hat{\theta})^\top Q(\theta-\hat{\theta})v_2
			- v_2^\top (\theta-\hat{\theta})^\top Q(\theta-\hat{\theta})v_2\\
			&\geq v_2^\top \hat{Q} v_2 - v_2^\top (\theta-\hat{\theta})^\top Q(\theta-\hat{\theta})v_2 \\
			&-v_2^\top(\theta - \hat{\theta})^\top Q (\theta - \hat{\theta})X\hat{R}^{-1}X^\top (\theta - \hat{\theta})^\top Q (\theta - \hat{\theta}) v_2.
		\end{split}
	\end{equation*}
	Next, we employ
	\begin{align*}
		&\gamma^2 \hat{R} \succeq R \succeq (\theta X - Y)^\top Q (\theta X - Y) \\
		\implies &\gamma ^2 I \succeq \hat{R}^{-\tfrac{1}{2}} (\theta X - Y)^\top Q (\theta X - Y) \hat{R}^{-\tfrac{1}{2}}
	\end{align*}
	to upper bound $\sigma_{\mathrm{max}}(\hat{R}^{-\tfrac{1}{2}} (\theta X - Y)^\top Q (\theta X - Y) \hat{R}^{-\tfrac{1}{2}})$.
	Using $\sigma_{\mathrm{max}}(A^\top A) = \sigma_{\mathrm{max}}(AA^\top)$ and $\hat{\theta}X = Y$, this yields
	\begin{align*}
		&-v_2^\top(\theta - \hat{\theta})^\top Q (\theta - \hat{\theta})X\hat{R}^{-1}X^\top (\theta - \hat{\theta})^\top Q (\theta - \hat{\theta})\\
		& \geq -\gamma^2 v_2^\top (\theta-\hat{\theta})^\top Q(\theta-\hat{\theta})v_2
	\end{align*}
	Lastly, employing this bound and $\hat{\theta}G_0=\bar{\theta}G_0$ yields
	\begin{equation*}
		v^\top G_0^\top (\hat{Q} - (\theta - \bar{\theta})^\top(Q+\gamma^2 Q) (\theta - \bar{\theta}))G_0 v \geq 0,
	\end{equation*}
	 which is satisfied by \eqref{eq:Main:QhatTheorem}. Hence, we showed that $\theta\in\hat{\Sigma}$.

	\ref{item:Main:TheoNotSet}) This follows from multiplying the \ac{QMI} for $\hat{\Sigma}$ with $X^\top$, and $X$ respectively, $GX = I$ and $G_0 X = 0$.
	
	\ref{item:Main:TheoConvergence}) First, we consider $\theta\not\in\Sigma$. 
	Statement \ref{item:Main:TheoNotSet}) and $\hat{R}\to0$ imply $\theta\not\in\hat{\Sigma}$ .
	Next, we consider $\theta\in\Sigma$ for a fixed $\hat{R}$. 
	Following the same steps as in \ref{item:Main:TheoSetInt}) leads to
	\begin{align*}
		& v_2^\top (\hat{Q} -  (\theta-\hat{\theta})^\top(Q+\gamma^2 Q)(\theta-\hat{\theta})) v_2 \\
		\geq  &\, (\sigma_{\mathrm{min}}(\hat{Q}) - \sigma_{\mathrm{max}}(Q + \gamma^2 Q)\sigma_{\mathrm{max}}(\theta-\hat{\theta})^2) \|v_2\|^2\geq 0,
	\end{align*}
	which is satisfied for $\sigma_{\mathrm{min}} (\hat{Q})\geq  \sigma_{\mathrm{max}}(Q + \gamma^2 Q)\sigma_{\mathrm{max}}(\theta-\hat{\theta})^2$. 
	Hence, for any $\theta\in\Sigma$ and $\hat{R}$, we find a sufficiently large $\hat{Q}$ such that $\theta\in\hat{\Sigma}$.	
\end{proof}

Note that even for the simple case, where $\theta$ is a column vector, finding the smallest ellipsoid containing $\Sigma \cap \Sigma_0$ is NP-complete \cite[Section 3.7.2]{Boyd1994}.
Thus, the previous Theorem establishes an approximate representation of $\Sigma$ using a \ac{QMI}.
Computing the intersection is performed as a second independent step using the S-Lemma.

For the special case, where $X$ is a square invertible matrix, we deduce $G_0= 0$, such that $\hat{Q}$ can be chosen arbitrarily.
Then $\Sigma\subset\hat{\Sigma}$ follows from $\hat{R}\succ 0$.
How to derive a suitable parameterization if $X$ has full row rank is addressed in \cite{Braendle2024}.

Theorem\,\ref{Theo:main:MainTheo} introduces $\hat{R}$, $\hat{Q}$, and $\bar{\theta}$ as decision variables.
To provide an intuition on how to interpret them, Fig.\,\ref{fig:main:TheoVisual} illustrates Theorem\,\ref{Theo:main:MainTheo} for a two-dimensional case with $\theta\in\mathbb{R}^{1 \times 2}$ and $Q=1$.
The unbounded set $\Sigma$ is shown in light gray.
The goal is to find an ellipse, that contains $\Sigma_{0} \cap \Sigma$, while not excessively restricting $\theta X$.
To this end, $G^\top\hat{R}G+G_0^\top\hat{Q}G_0$ is introduced with $G$ and $G_0$ isolating the column space of $X$ and its complement in $\mathbb{R}^n$, respectively. 
To gain a more intuitive understanding, each term can be interpreted as the radius of a ball in the corresponding vector space.
$\hat{Q}$ is the radius of the unbounded component of $\theta$, i.e., any part of $\theta$ satisfying $\theta X = 0$.
Now, $\hat{Q}$ must be chosen sufficiently large to contain $\Sigma_{0} \cap \Sigma$, which is generalized in \eqref{eq:Main:QhatTheorem} to achieve \ref{item:Main:TheoSetInt}).
$\hat{R}$ on the other hand increases the radius for the row space of $\theta$ that coincides with $X$, thereby increasing the bounded component.
This means that there may exist $\theta\in\hat{\Sigma}$ even though $\theta\not\in\Sigma$.
This additional radius is used to find an outer approximation of $\Sigma_{0} \cap \Sigma$ and must be chosen as a trade-off. 
Choosing it very small, results in a tight description close to $\Sigma$, but doing so requires a larger $\hat{Q}$.
This can lead to poor conditioned matrices and can affect the approximation of the intersection.

\begin{figure}[t!]
	\centering
	\begin{tikzpicture}[scale=0.6]
	
	\newcommand{\lineann}[4][0.5]{%
		\begin{scope}[rotate=#2, inner sep=2pt]
			\draw[dashed] (0,0) -- +(0,#1)
			node [coordinate, near end] (a) {};
			\draw[dashed] (#3,0) -- +(0,#1)
			node [coordinate, near end] (b) {};
			\draw[|<->|] (a) -- node[rotate = #2, fill=white] {#4} (b);
		\end{scope}
	}
	
	\newcommand{\lineannShort}[4][0.5]{%
		\begin{scope}[rotate=#2, inner sep=2pt]
			\coordinate (a) at (0,0);	
			\coordinate (b) at (#3,0);	
			\draw[dashed] (b) -- +(0,#1);
			\draw[<->] (a) -- node[rotate = #2, above,align=center] {#4} (b);
		\end{scope}
	}

	\clip (-5.1,-4) rectangle (5,3.2);
	
	\node at (0,0) [above right]{$\hat{\theta}$};
	\fill (0,0) circle(1.5pt);
	\filldraw[draw = black, fill opacity=0.2, fill=gray] (-10,-2) rectangle ++(20,4);
	\node at (4.0,2.0) [above] {$\Sigma$};
	
	\draw[ decoration={brace,	mirror,	raise=1	},decorate] (-0.0,-0)--(-0.0,-2);
	\node at (-1.2,-1) [] {\small{$G^\top R G$}};

	\draw[ decoration={brace,	mirror,	raise=1	},decorate] (0,0)--(4.5,0);
	\node at (2.25,-0.1) [below] {\small{$G_0^\top \hat{Q} G_0$}};
	
	
	\draw[->] (-4.8,0) -- (4.8,0);
	\draw [->] (0,-3.0)-- (0,3.0);
	
	\draw (1.3,-2.5) circle(1.5);
	\node at (2.5,-4) [above right]{$\Sigma_{0}$};
	
	\draw (0,0) ellipse(4.5 and 2.5);
	\node at (0.5,2.4) [above right]{$\hat{\Sigma}$};
	
	\draw[ decoration={brace,	mirror,	raise=1	},decorate] (-0,2.5) --  (0,0);
	\node at (-1.9,1.25) [] {\small{$G^\top(R+\hat{R})G$}};
	\node at (-1.2,-1) [] {\small{$G^\top R G$}};
	
\end{tikzpicture}
	\caption{Visualization of $\Sigma$, $\Sigma_{0}$ and $\hat{\Sigma}$.} \label{fig:main:TheoVisual}
\end{figure}
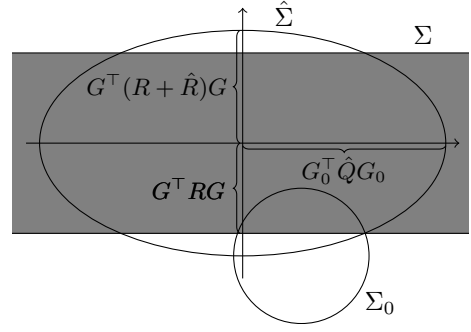

Statement \ref{item:Main:TheoConvergence}) of Theorem\,\ref{Theo:main:MainTheo} provides convergence guarantees for the limits of $\hat{R}$ and $\hat{Q}$.
As illustrated in Fig.\,\ref{fig:main:TheoVisual}, choosing $\hat{R}$ sufficiently small and $\hat{Q}$ sufficiently large, ensures every $\theta\in\Sigma$ is also contained in $\hat{\Sigma}$.

Finally, $\bar{\theta}$ can be chosen arbitrarily.
However, as shown in Fig.\,\ref{fig:main:TheoVisual}, if $\hat{\theta}=YG + \bar{\theta}G_0$ is far away from $\Sigma_{0}$, then a larger $\hat{Q}$ is required to compensate for the difference.
This can play a role when approximating the intersection of multiple ellipsoids using techniques like the S-Lemma.
A heuristic choice for $\bar{\theta}$ is the center of $\Sigma_{0}$.
Next, we provide a way to compute a suitable $\hat{Q}$.
\begin{corollary}\label{coro:Main:Compuation}
	Suppose $\bar{\theta}\in\mathbb{R}^{p \times n}$ and  $\epsilon > 0$. If
	\begin{equation*}
		\begin{bmatrix}
			I & - \bar{\theta} \\
			0 & I
		\end{bmatrix}^\top
		\begin{bmatrix}
			-(Q+\gamma^2Q) & 0 \\
			0 & \hat{Q}
		\end{bmatrix}
		\begin{bmatrix}
			I & - \bar{\theta} \\
			0 & I
		\end{bmatrix} - \sum_{i=1}^{N_0} \lambda_{\mathrm{i}}\Phi_\mathrm{0,i} \succeq 0 
	\end{equation*}
	with $\lambda_{\mathrm{i}}\geq0$ $i=1,\ldots,N_0$, $\hat{Q} \succeq 0$, $\hat{R} = ((1+\epsilon)^2-1)^{-1}R$  and $\gamma^2 = ((1+\epsilon)^2-1)^{-1}$, then \eqref{eq:Main:RhatTheorem} and \eqref{eq:Main:QhatTheorem} are satisfied.
\end{corollary}
\begin{proof}
	\eqref{eq:Main:RhatTheorem} follows from the definition of $\hat{R}$ and \eqref{eq:Main:QhatTheorem} follows from the S-Lemma and enforcing positive definiteness of \eqref{eq:Main:QhatTheorem} for the whole $\mathbb{R}^n$ not only for $G_0$.
\end{proof}
Corollary\,\ref{coro:Main:Compuation} describes a convex \ac{SDP}, which can be solved efficiently. 
Furthermore, it allows for defining an objective function, such as minimizing $\mathrm{tr}(\hat{Q})$ or $\mathrm{log\,det}(\hat{Q})$ to reduce the size of $\hat{\Sigma}$.
The choice $\hat{R} = (1+\epsilon)^2-R$ leads to $R+\hat{R}=(1+\epsilon)^2 R$.
By comparing this with Theorem\,\ref{Theo:main:MainTheo}\,\ref{item:Main:TheoNotSet}), we conclude that $\epsilon$ is the maximum allowable relative deviation from $\Sigma$.
Thus, we first select a sufficiently small $\epsilon$, then solve for a suitable $\hat{Q}$ to apply Theorem\,\ref{Theo:main:MainTheo}. 
Note, Corollary\,\ref{coro:Main:Compuation} is independent of $X$, meaning the same $\hat{Q}$ and $\hat{R}$ can be reused for multiple samples, reducing the computational effort.

\subsection{Prior Knowledge} \label{subsec:Prior}
In this subsection, we present a method to incorporate structural constraints.
Data-informativity methods typically treat the system as a black box.
However, some additional knowledge may be available, such as a possible range of a parameter.
We consider constraints of the form
\begin{equation}
	\sigma_{\mathrm{max}}(E\Theta  F + G)^2 \leq \epsilon^2.
\end{equation}
For example, in the scalar case, appropriate $E$ and $F$ allow to extract any individual entry of $\theta$.
If $E$ has full row rank, the constraint can be reformulated as
\begin{equation}
	\begin{bmatrix}
		\Theta F + E^\dagger G \\I
	\end{bmatrix}^\top
	\begin{bmatrix}
		-E^\top E& 0 \\
		0 & \epsilon^2 I
	\end{bmatrix}
	\begin{bmatrix}
		\Theta F + E^\dagger G \\I
	\end{bmatrix}
	\succeq 0.
\end{equation}
By applying Theorem\,\ref{Theo:main:MainTheo}, we can derive additional \acp{QMI} for the set membership estimation framework.

\subsection{Multiple Datapoints} \label{subsec:MultPoints}
A significant limitation of Theorem\,\ref{Theo:main:MainTheo} is the requirement for a prior $\Sigma_{0}$.
Depending on the problem, it may be reasonable to assume access to such prior knowledge.
However, if not, it is possible to construct $\Sigma_{0}$ from the collected data.
\begin{lemma} \label{lemma:main:EnergyBound}
	Suppose $Q\succeq0$ and $R\succeq0$. If $d_i\in\mathbb{R}^{p\times m}$ satisfies
	\begin{equation}
		\begin{bmatrix}
			d_i \\ I
		\end{bmatrix}^\top
		\begin{bmatrix}
			-Q & 0 \\
			0 & R
		\end{bmatrix}
		\begin{bmatrix}
			d_i \\ I
		\end{bmatrix} \succeq 0
	\end{equation}
	for $i=1,\dots,N$, then
	\begin{equation}
		\begin{bmatrix}
			D \\ I
		\end{bmatrix}^\top
		\mathrm{diag}(-Q,NR,\ldots,NR)
		\begin{bmatrix}
			D \\ I
		\end{bmatrix} \succeq 0
	\end{equation}
	with $D=\begin{bmatrix}d_1 & d_2 & \ldots & d_N\end{bmatrix}$.
\end{lemma}
\begin{proof}
	Consider an arbitrary partitioned vector $v^\top = [v_1^\top ,\ldots,v_N^\top ]$. 
	Multiplication from the left and right by $v$ and $v^\top$ respectively, yields
	\begin{align*}
		&\sum_{i=1}^{N} \! v_i^\top \!(R-d_i^\top \!Q d_i) v_i + (N-1)v_i^\top \! Rv_i- \!\sum_{\substack{i,j=1\\i\neq j }}^{N}v_i^\top \! d_i^\top Q d_j v_j \\
		&\geq \sum_{i=1}^{N} \! v_i^\top \!(R-d_i^\top \!Q d_i) v_i + (N-1)v_i^\top \! (R-d_i^\top Q d_i)v_i.
	\end{align*}
	by the following inequality
	\begin{equation*}
		v_i^\top d_i^\top Q d_i v_i + v_j^\top d_j^\top Q d_j v_j \geq 2 v_i^\top d_i^\top Q d_j v_j.
	\end{equation*}
	Finally, by assumption it holds
	\begin{equation*}
		R - d_i^\top Q d_i \succeq 0 \qquad i = 1,\ldots,N,
	\end{equation*}
    which concludes the proof.
\end{proof}
Similar results can be found in literature for $Q=I$ \cite{Berberich2020, Waarde2022}.
To apply Lemma\,\ref{lemma:main:EnergyBound}, we rewrite \eqref{eq:Main:DataEq} in terms of the horizontally stacked data matrices
\begin{equation}
	\begin{bmatrix}
		Y_1,\ldots,Y_N
	\end{bmatrix}=
	\theta
	\begin{bmatrix}
		X_1,\ldots,X_N
	\end{bmatrix}+
	\begin{bmatrix}
		W_1,\ldots,W_N
	\end{bmatrix}
\end{equation}
and get a corresponding \ac{QMI} bound for $\begin{bmatrix}W_1,\ldots,W_N\end{bmatrix}$.
If $\begin{bmatrix}X_1,\ldots,X_N\end{bmatrix}$ has full row rank, the results in \cite[Theorem 1]{Braendle2024} can be used to compute $\Sigma_{0}$.
Although, Lemma\,\ref{lemma:main:EnergyBound} provides an approximation of the exact set intersection and, thus introduce some conservatism, it serves as a useful prior and can be intersected with other available priors to get an even tighter description.

After constructing a suitable $\Sigma_{0}$, we can compute $\Sigma_{0}\cap\bigcap_{k=1}^{N} \Sigma_k$.
Since this problem is NP-complete, we seek only an outer approximation \cite[Section 3.7.2]{Boyd1994}.
As commonly done in the data-informativity framework \cite{Waarde2022}, we employ the S-Lemma to parameterize an approximation
\begin{equation*}
	\begin{split}
		\bar{\Sigma} = \{\theta\in\mathbb{R}^{p \times n} \mid [\star]^\top \left(\sum_{i=1}^{N_0} \tau_i \Phi_{0,i} + \sum_{k = 1}^{N} \lambda_k \Phi_{k}\right) \begin{bmatrix}		\theta\\I\end{bmatrix} \succeq 0 \\
		\lambda_i \geq 0 \quad \tau_k\geq 0 \qquad \forall i = 1,\ldots, N_0 \quad k=1,\ldots,N\}
	\end{split}
\end{equation*}
with $\Phi_{k}$ being the corresponding matrix of $\hat{\Sigma}_k$ and $\tau$ and $\lambda$ being decision variables.
This results in an affine outer approximation of the intersection.

\subsection{Comparison Data-Informativity}
In this section, we want to compare our parameterization to the data-informativity framework. 
To do so, we consider the special case, $Q\succ 0$, $\hat{Q}\succ0$ and $R\succ 0$.
First, state a lemma to equivalently transform between these frameworks.
\begin{lemma} \label{lemma:Main:Dual}
	Suppose $Q\succ 0$ and $R\succ 0$, then
	\begin{equation}
		\begin{aligned}
		&[\star]^\top
		\begin{bmatrix}
			-Q & 0 \\0 & R
		\end{bmatrix}
		\begin{bmatrix}
			D \\ I
		\end{bmatrix} \succ 0 \\
		\iff
		&[\star]^\top
		\begin{bmatrix}
			-Q^{-1} & 0 \\0 & R^{-1}
		\end{bmatrix}
		\begin{bmatrix}
			I \\ D^\top
		\end{bmatrix} \prec 0
		\end{aligned}
	\end{equation}
\end{lemma}
\begin{proof}
	This is a special case of the dualization lemma as presented in \cite[Lemma 4.9]{Weiland1994}
\end{proof}
To apply Lemma\,\ref{lemma:Main:Dual}, we first consider $\hat{\Sigma}$.
We note that it holds 
\begin{equation*}
	G_0 = (I-XG) = \varphi G^\perp 
\end{equation*}
with $\varphi\in\mathbb{R}^{n \times p-n}$ being a matrix with full column rank to isolate the row space of $G_0$.
Hence, it follows
\begin{align*}
	&\begin{bmatrix}
		-Q & 0\\
		0 & G^\top(R+ \hat{R})G + G_0^\top \hat{Q} G_0
	\end{bmatrix}^{-1}
	\\&=
	\begin{bmatrix}
		-Q^{-1} & 0 \\
		0& X(R+\hat{R})^{-1} X^\top + G^{\perp\top} (\varphi^\top \hat{Q}\varphi)^{-1}G^{\perp}
	\end{bmatrix},
\end{align*}
from the projection properties of $G_0$. 
The inverse $(\varphi^\top \hat{Q}\varphi)^{-1}$ exists, because $\varphi$ has full column rank.
To compare it to the data-informativity framework, we first have to apply Lemma\,\ref{lemma:Main:Dual} on $\mathcal{W}$ and insert
$W^\top = Y^\top - X^\top \theta^\top$ to get
\begin{equation*}
	[\star]^\top \begin{bmatrix}
		-Q^{-1} & 0\\
		0 & XR^{-1}X^\top
	\end{bmatrix}
	\begin{bmatrix}
		I\\
		\theta^\top - G^\top Y^\top - G_0^\top \bar{\theta}
	\end{bmatrix} \prec 0
\end{equation*}
with $X^\top G_0^\top=0$.
As already stated in Theorem\:\ref{Theo:main:MainTheo}\ref{item:Main:TheoConvergence}) for $\sigma_{\mathrm{max}}(\hat{R})\to0$ and $\sigma_{\mathrm{min}}(\hat{Q})\to\infty$, we get an exact description of $\Sigma$, as also provided by the data-informativity framework.
In addition, this offers a similar parameterization, how to modify the \ac{QMI} in the data-informativity framework, in order to apply dualization. 
Instead of adding a sufficiently large constant matrix, one must add a sufficiently small matrix on the term weighting $\theta^\top$.

Since our method is closely related to the data-informativity framework and the key difference is that our formulation is a \ac{QMI} in $\theta$ instead of $\theta^\top$, we now highlight the differences and advantages of each approach.
Both methods rely on outer approximations to describe the set of all systems that are consistent with the collected data.
Ellipsoidal approximations are used, because they lead to convex problems, when applying robust control methods.
The data-informativity framework formulates the \ac{QMI} in terms of $\theta^\top$, which is well suited for robust state-feedback synthesis, since it allows for simultaneous optimization over the controller and an affine parameterization of the uncertainty as provided by the S-Lemma \cite[Section 8.3]{Weiland1994}.
Our approach, on the other hand, parameterizes the problem directly as \ac{QMI} in $\theta$, offering similar benefits but for the estimator synthesis allowing to optimize over $\tau$ and $\lambda$ in $\bar{\Sigma}$.

Under certain invertibility assumptions, it is possible to equivalently transform between $\theta^\top$ and $\theta$ through dualization \cite[Section 4.4.1]{Weiland1994}.
When solving the corresponding \ac{SDP}, there exist some optimal ellipsoidal approximation.
If it is invertible, then $\theta^\top$ can be translated to $\theta$ by applying dualization.
However, to optimize over the ellipsoid, an affine parameterization is needed.
Even for vectors finding the smallest ellipsoid is NP-complete and hence not easily solvable.
Instead, the S-Lemma is commonly used to derive an affine parameterization of the intersection\cite{Waarde2022}. 
This comes at the cost of conservatism, such that it may not be possible to represent the optimal ellipsoid.
Moreover, when taking the inverse, any affine parameterization is lost and must be fixed beforehand \cite[Proposition 2]{Martin2022}.
This leads to additional conservatism.
Specifically, there may not even exist $\lambda \geq 0$, such that 
\begin{equation}
	\left(\sum_{i=1}^{N}\lambda_i\Phi_i^{-1}\right)^{-1} = \sum_{i=1}^{N} \tau_i \Phi_i \label{eq:Param:DualizationNotEquiv}
\end{equation} 
for every $\tau\geq0$. 
In fact, this can lead to a completely different set, depending whether one works with $\theta$ or $\theta^\top$. 

This is the key motivation behind our approach.
By working in the desired space, dual or primal, we ensure that any affine parameterization is kept, reducing conservatism.
Moreover, Section\,\ref{subsec:Prior} presents an additional structure, for which dualization is not applicable. 
Nonetheless, we derive a suitable \ac{QMI}.
We have to stress again, that due to \eqref{eq:Param:DualizationNotEquiv}, the parameterizations differ.
Hence, we see both descriptions as complementary rather than mutually exclusive and using both approaches in combination may yield better results, than each on their own.

\section{Estimator Synthesis}\label{sec:Param}
The goal of this section is to design an estimator for the following discrete-time \ac{LTI} system of the form
\begin{equation}
	\begin{bmatrix}
		x_{k+1} \\ \cmidrule(lr){1-1}
		z_{\mathrm{p},k} \\
		y_k
	\end{bmatrix} = 
	\left[\begin{array} {c|c}
		A_{\mathrm{tr}} & B_{\mathrm{p,tr}} \\ \hline
		C_{\mathrm{p}} & D_{\mathrm{p}} \\
		C_{\mathrm{y,tr}} & D_{\mathrm{yp,tr}}
	\end{array}\right]
	\begin{bmatrix}
		x_{k} \\ \cmidrule(lr){1-1}
		w_{\mathrm{p},k} 
	\end{bmatrix}
	+ \begin{bmatrix}
		w_k \\ \cmidrule(lr){1-1} 0 \\ v_k 
	\end{bmatrix}
\end{equation}
where $x_k\in\mathbb{R}^{n_x}$ is the state of the system, $w_{\mathrm{p}}\in\mathbb{R}^{m_\mathrm{p}}$ is a performance input, $y\in\mathbb{R}^{p_y}$ is a measurement signal and $z_\mathrm{p}\in\mathbb{R}^{p_p}$ is the signal to be estimated. 
The matrices $A_{\mathrm{tr}}$, $B_{\mathrm{p,tr}}$, $C_{y,\mathrm{tr}}$ and $D_{\mathrm{yp,tr}}$ are unknown.
The matrices $C_{\mathrm{p}}$ and $D_{\mathrm{p}}$ are assumed to be known, but the following steps can be extended to the case, where they are also unknown.
The signals $w_k$ and $v_k$ are unknown, additive disturbances affecting the data collection and satisfy a \ac{QMI} with known $Q_w\succeq0$, $R_w\succeq0$, $Q_v\succeq0$, and $R_v \succeq 0$ as in \eqref{eq:main:WSet}.
Since each instance of $w_k$ and $v_k$ are bounded for each time step, this corresponds to the instantaneous noise case.
The estimator is also an \ac{LTI} system of the following form
\begin{equation} \label{eq:Param:EstimatorForm}
	\begin{bmatrix}
		\hat{x}_{k+1} \\ \cmidrule(lr){1-1}
		\hat{z}_{\mathrm{p},k} 
	\end{bmatrix} = 
	\left[\begin{array} {c|c}
		A_{\mathrm{E}} & B_{\mathrm{E}} \\ \hline
		C_{\mathrm{E}} & D_{\mathrm{E}} 
	\end{array}\right]
	\begin{bmatrix}
		\hat{x}_{k} \\ \cmidrule(lr){1-1}
		y_{k} 
	\end{bmatrix}.
\end{equation}
The objective is to minimize the $\mathcal{H}_\infty$-norm from the input $w_{\mathrm{p}}$ to the estimation error $z_{\mathrm{p}}- \hat{z}_{\mathrm{p}}$.

We consider a two-stage data collection process as in \cite{Mishra2022}.
In the first stage, we have access to $N+1$ measurements of the full state $\{x^d_k,x^d_{k+1}\}_{k=1}^{N}$ and $N$ measurements of the input and output $\{y^d_k,w^d_{\mathrm{p},k}\}_{k=0}^{N-1}$.
The second phase is the application of the synthesized estimator.
Here, we do not have access to the full state and input measurements anymore. 
Nonetheless, we still want to estimate $z_\mathrm{p}$ using only $y$.
Now, we can apply the method described in Section\,\ref{sec:MainTheorem} for $\begin{bmatrix} A,\,B_\mathrm{p}\end{bmatrix}$ and $\begin{bmatrix} C_\mathrm{y},\,D_\mathrm{yp}\end{bmatrix}$ to construct an affine parameterization using \acp{QMI}.
We denote the corresponding matrices with $\Phi_{AB_\mathrm{p}}$ and $\Phi_{C_\mathrm{y}D_\mathrm{yp}}$, respectively.
Thus, the set of all systems consistent with the collected data is described by
\begin{align*}
	&\begin{bmatrix}
		x_{k+1}  \\ \cmidrule(lr){1-1} z_k \\ z_{\mathrm{p},k} \\ y_k
	\end{bmatrix} = 
	\left[\begin{array}{c|cc}
		0 & \left[I_{n}\;0\right] & 0\\ \hline
		\begin{bmatrix} I_n \\ 0	\end{bmatrix} & 0 & \begin{bmatrix} 0 \\ I_{m_\mathrm{p}}	\end{bmatrix} \\
		C_\mathrm{p} & 0 & D_\mathrm{p} \\
		0 & \left[0\;I_{p_y}\right]  & 0  
	\end{array}\right]
	\begin{bmatrix}
		x_{k} \\ \cmidrule(lr){1-1} w_k \\ w_{\mathrm{p},k}
	\end{bmatrix}, \\
	&w_k = \Delta z_k, \qquad 
	[\star]^\top P 
	\begin{bmatrix}
		\Delta \\ I 
	\end{bmatrix}\succeq 0, \qquad
	\Delta = \begin{bmatrix}A & B_\mathrm{p} \\ C_{\mathrm{y}} & D_{\mathrm{yp}}\end{bmatrix},
\end{align*}
with
\begin{align*}
	P = &[\star]^\top \!\!\Phi_{AB_\mathrm{p}} \!\!\begin{bmatrix} I_n & \!\!0 & 0 \\ 0 & \!\!0 &\!\! I_{n+m_p}\! \end{bmatrix}+ 
	[\star]^\top \!\Phi_{C_\mathrm{y}D_\mathrm{yp}} \!\!\begin{bmatrix} 0 & \!\!\! I_{p_y} & 0 \\ 0 & \!\!0 & \!\! I_{n+m_p}\!\end{bmatrix}\!.
\end{align*}
This reformulation, together with \cite[Theorem 4]{Braendle2024}, enables the synthesis of an estimator with a guaranteed upper bound on $\mathcal{H}_\infty$-norm of the true system.
The resulting problem is a convex \ac{SDP}, which can solved be efficiently.
Note, that this requires an uncertainty parameterization in terms of $\theta$ instead of $\theta^\top$ and allows to simultaneously optimize for the $\mathcal{H}_\infty$-norm, the estimator and any affine parameterization of $P$.



\section{Numerical Example}

To validate the theoretical analysis, we apply Theorem\,\ref{Theo:main:MainTheo} to two simple \ac{LTI} systems.
For this purpose, we generate data $\{x^d_k,x^d_{k+1},y^d_k,w^d_{\mathrm{p},}\}_{k=1}^{N}$ with instantaneous noise $w_k$ and $v_k$ with $Q=I$ and $R=\alpha^2I$. 
The same data set is used for all experiments to ensure comparability.
In all examples, Corollary\,\ref{coro:Main:Compuation} is used to determine $\hat{Q}$, while simultaneously minimizing $\mathrm{trace}(\hat{Q})$.
All semidefinite programs are solved in Matlab using the toolbox YALMIP \cite{Lofberg2004} and the solver MOSEK \cite{mosek}.
Moreover, we chose $\epsilon=0.1$ to allow a maximum relative deviation of $10\%$.
For comparison, we consider four different priors.
$\Sigma_{0,D}$ is the combined prior using all the available data samples as described in Section \ref{subsec:Prior} and \cite{Braendle2024}.
As second prior $\Sigma_{0,L}$, we choose $\sigma_{\mathrm{max}}([A\,B_\mathrm{p}]-1.04 [A_{\mathrm{tr}}\,B_\mathrm{p,tr}])\leq \beta^2$ and $\sigma_{\mathrm{max}}([C_{\mathrm{y}}\,D_\mathrm{yp}]-[C_{\mathrm{y,tr}}\,D_\mathrm{yp,tr}])\leq \beta^2$.
The prior $\Sigma_{0,D}$ is computed using the data-informativity framework and \cite[Proposition 2]{Martin2022} to compute an intersection in order to apply dualization.
Lastly, we combine all priors in $\Sigma_{0,C}$.
We always choose $\bar{\theta}$ as the center of the corresponding \ac{QMI}, except for $\Sigma_{0,D}$, for which we chose the true value.
For each of the priors we apply Theorem\,\ref{Theo:main:MainTheo} to determine additional \acp{QMI} based on each sample individually.
Further, $\emptyset$ denotes the case without any instantaneous bounds to serve as a baseline. 
Next, $\Sigma_D$ contains the instantaneous bounds based on the collected data.
The set $\Sigma_P$ represents additional prior knowledge as described in Section \ref{subsec:Prior}.
Lastly, $\Sigma_C$ is the combination of $\Sigma_D$ and $\Sigma_P$. 
For all possible combinations, we synthesize an estimator for the full-state and compare the guaranteed upper bound on the $\mathcal{H}_\infty$-norm.
\begin{table}[t!]
	\caption{Guaranteed upper bounds on the $\mathcal{H}_\infty$ for system 1.}
	\begin{center}
		\label{tab:2DSystem}
		\begin{tabular}{c|cccc} 
			& $\Sigma_{0,D}$ & $\Sigma_{0,L}$ & $\Sigma_{0,I}$ & $\Sigma_{0,C}$\\ \hline
			$\emptyset$ &        $2.791$ &        $22.65$ &        $3.423$ &        $2.791$\\
			$\Sigma_D$  &        $2.791$ &        $6.429$ &        $3.423$ &        $2.791$\\
			$\Sigma_P$  &        $2.791$ &       $17.794$ &        $3.423$ &        $2.791$\\
			$\Sigma_C$  &        $2.791$ &        $6.429$ &        $3.423$ &        $2.791$
		\end{tabular}
	\end{center}
\end{table}
First, we consider
\begin{align*}
	&A_\mathrm{tr} = \begin{bmatrix} 0.7 & 0 \\ 0.3 & 0.7\end{bmatrix},
	& B_\mathrm{p,tr} = \begin{bmatrix} 1 & 0 \\ 0 & 0\end{bmatrix}, \\
	&C_\mathrm{y,tr} = \begin{bmatrix} 0 & 1\end{bmatrix},
	& D_\mathrm{yp,tr} = \begin{bmatrix} 0 & 1\end{bmatrix}.
\end{align*}
We generate $N=50$ samples, with $\beta = 0.15$ and $\alpha=0.0005$. 
Additionally, we sample a uniformly distributed initial condition $x_0\in[-2,2]^2$ and performance input $w_{\mathrm{p},k}\in[-2,2]^2$.
All zero entries are enforced to have absolute values less than $0.01$, except for the first column of $B_\mathrm{p}$.
Additionally, we enforce the last row of $A$ to sum up to $1$ as does the first column and we enforce the entry in the second row of the first column to be within $1\%$ of the true value as additional prior knowledge.

The guaranteed upper bounds on the $\mathcal{H}_\infty$-norm are summarized in Table\,\ref{tab:2DSystem}.
Interestingly, for this example, we see that $\Sigma_{0,\mathrm{D}}$ gives the best results, despite relying on the conservative Lemma\,\ref{lemma:main:EnergyBound}.
This shows that computing the intersection using the S-Lemma is another source of conservatism.
In contrast, when considering $\Sigma_{0,\mathrm{L}}$, we see a significant improvement when adding with $\Sigma_{\mathrm{D}}$ and $\Sigma_{\mathrm{P}}$, since $\Sigma_{0,\mathrm{L}}$ is chosen to be larger than the other priors.
This is expected, since the added constraints are rather tight, again highlighting the conservatism due to the S-lemma.

As a second example, we consider
\begin{align*}
	&A_\mathrm{tr} = 0.8,
	&B_\mathrm{p,tr} = 1, \qquad 
	&C_\mathrm{y,tr} = 1,
	&D_\mathrm{yp,tr} = 0.1,
\end{align*}
with $N = 10$, $\alpha = 0.6$, and $\beta = 0.1$.
The initial condition is sampled uniformly within $[-10,10]$ and $w_{\mathrm{p},k}$ in $[-4,6]$.
We enforce $B_\mathrm{p,tr}$ to be within $1\%$ of the true value.
The results are illustrated in Table\,\ref{tab:1DSystem}.
For this example, we note that including $\Sigma_{\mathrm{D}}$ improves the achieved $\mathcal{H}_\infty$-norm of the estimation error.
Moreover, we observe that combining the priors also improves the $\mathcal{H}_\infty$-norm.
Interestingly $\Sigma_{\mathrm{D}}$ leads to a significant improvement when combined with $\Sigma_{\mathrm{0,L}}$.

These examples have shown, that including additional information in the form of $\Sigma_D$ or $\Sigma_P$ can improve performance, especially compared to the data-informativity framework. 
However, due to the conservative nature of the S-lemma, an improvement can not be guaranteed in general.
\begin{table}[t!]
	\centering
		 \caption{Guaranteed upper bounds on the $\mathcal{H}_\infty$ for system 2.}
		\label{tab:1DSystem}
		\begin{tabular}{c|cccc}
			& $\Sigma_{0,D}$ & $\Sigma_{0,L}$ & $\Sigma_{0,I}$ & $\Sigma_{0,C}$\\ \hline
			$\emptyset$ &        $0.678$ &        $1.53$  &        $1.357$ &        $0.649$\\
			$\Sigma_D$  &        $0.525$ &        $0.428$ &        $1.197$ &        $0.434$\\
			$\Sigma_P$  &        $0.678$ &        $1.53$  &        $1.281$ &        $0.649$\\
			$\Sigma_C$  &        $0.525$ &        $0.428$ &        $1.197$ &        $0.434$
		\end{tabular}
\end{table}

\section{Conclusion}
In this paper, we presented new parameterization for performing set-membership estimation.
Our approach allows for the simultaneous optimization of an estimator and additional multipliers, thereby improving performance guarantees, while avoiding the conservatism introduced by dualization techniques.
Nonetheless, the conservative nature of S-lemma still leads to sub-optimal results.

\bibliographystyle{IEEEtran}
\bibliography{CDC_Literature}

\end{document}